\documentclass[12pt, draftclsnofoot, onecolumn]{IEEEtran}

\usepackage{setspace}
\usepackage{amsmath}
\usepackage{amsfonts}
\usepackage{graphicx}
\usepackage{amsthm}
\usepackage{amssymb}
\usepackage{mathrsfs}
\usepackage{stfloats}
\usepackage{geometry}
\usepackage{color}
\usepackage{bm}

\newtheorem{lemma}{Lemma}
\newtheorem{theorem}{Theorem}

\begin{document}

\title{Secure Massive MIMO Communication with Low-resolution DACs}


\author{\small
\IEEEauthorblockN{
Jindan Xu$^1$,~\emph{Student Member,~IEEE},
Wei Xu$^1$,~\emph{Senior Member,~IEEE},
Jun Zhu$^2$,~\emph{Member,~IEEE}}, \\
Derrick Wing Kwan Ng$^3$,~\emph{Senior Member,~IEEE},
and A. Lee Swindlehurst$^4$,~\emph{Fellow,~IEEE}
\\
\IEEEauthorblockA{
$^1$National Mobile Communications Research Laboratory, Southeast University, Nanjing 210096, China\\
$^2$Qualcomm Inc. 5775 Morehouse Drive, San Diego, CA 92121, USA\\
$^3$School of Electrical Engineering and Telecommunications, University of New South Wales, NSW 2052, Australia\\
$^4$Center for Pervasive Communications and Computing, University of California, Irvine, CA 92697, USA\\
Email: \{jdxu, wxu\}@seu.edu.cn, junzhu@qti.qualcomm.com, w.k.ng@unsw.edu.au, swindle@uci.edu}
}

\maketitle

\begin{abstract}
In this paper, we investigate secure transmission in a massive multiple-input multiple-output (MIMO) system adopting low-resolution digital-to-analog converters (DACs).
Artificial noise (AN) is deliberately transmitted simultaneously with the confidential signals to degrade the eavesdropper's channel quality.
By applying the Bussgang theorem, a DAC quantization model is developed which facilitates the analysis of the asymptotic achievable secrecy rate.
Interestingly, for a fixed power allocation factor $\phi$, low-resolution DACs typically result in a secrecy rate loss, but in certain cases they provide superior performance, e.g., at low signal-to-noise ratio (SNR).
Specifically, we derive a closed-form SNR threshold which determines whether low-resolution or high-resolution DACs are preferable for improving the secrecy rate.
Furthermore, a closed-form expression for the optimal $\phi$ is derived.
With AN generated in the null-space of the user channel and the optimal $\phi$, low-resolution DACs inevitably cause secrecy rate loss.
On the other hand, for random AN with the optimal $\phi$, the secrecy rate is hardly affected by the DAC resolution because the negative impact of the quantization noise can be compensated for by reducing the AN power.
All the derived analytical results are verified by numerical simulations.



\end{abstract}


\IEEEpeerreviewmaketitle

\section{Introduction}
Secrecy plays an important role in wireless communications since it is difficult for a broadcast channel to shield transmit signals from unintended recipients.
Traditionally, secure transmission relies on key-based cryptographic methods implemented at the network and application layers \cite{Introduction}.
However, these cryptographic measures are based on the assumption that it is computationally infeasible for the encrypted message to be deciphered within a reasonable amount of time.
Consequently, they inevitably become more vulnerable as the computational capability of the adversary grows.
In the past decade, physical layer security, as a complement to existing cryptographic methods, has gained increasing attention \cite{Sec_01}-\cite{Sec_WPCN}.
With appropriate designs, physical layer techniques enable secure communication over a wireless medium without the help of encryption keys \cite{Sec_1}-\cite{Sec_3}.
In addition, they can be used to augment already existing security measures at higher layers, leading to a multilayer secure transmission \cite{Survey}.

The classical three-terminal security model, known as the wiretap channel, was originally proposed in \cite{Secure_0}, consisting of a transmitter (Alice), an intended receiver (Bob), and an unauthorized receiver (Eve) referred to as an eavesdropper.
This concept has been extended to multi-antenna networks \cite{Sec_MIMO_1}, \cite{Sec_Yu}, while beamforming techniques have been utilized in multiple-input multiple-output (MIMO) systems to improve secrecy \cite{Sec_MIMO_2}.
When the instantaneous channel state information (CSI) of the eavesdropper is known at the transmitter, it has been demonstrated in \cite{Sec_MIMO_GSVD} that the generalized singular value decomposition (GSVD) precoding scheme can achieve the secrecy capacity in the high signal-to-noise ratio (SNR) limit.
The study in \cite{Sec_Korner} showed that secret communication is possible if the eavesdropper's channel is more noisy than the user channel.
When the eavesdropper happens to have a better channel than the legitimate user (e.g., if the eavesdropper is much closer to the transmitter), artificial noise (AN) has been proposed in \cite{Sec_MIMO_AN} and \cite{Sec_MIMO_AN_3} to help degrade the channel quality of the eavesdropper.
The AN is usually designed to be orthogonal to the channel of the intended receivers, thus causing no additional interference to the legitimate users \cite{Sec_MIMO_AN_1}, \cite{Sec_MIMO_AN_2}.
In order to further combat the uncertainty of channel information at the transmitter, robust beamforming design for physical layer security with the aid of AN has been studied in \cite{Sec_BDMA}.

Recently, massive MIMO has become a candidate technology for next-generation wireless communication systems \cite{MIMO1}-\cite{5G} and its application to guarantee communication security has attracted significant attention.
In massive MIMO, hundreds, or even thousands, of antennas are equipped at the base station (BS) \cite{MIMO3}-\cite{MIMO_Liu} and the corresponding spatial-wideband effect has been studied in \cite{MIMO_SWB}.
For instance, downlink secure transmission at the physical layer in a multi-cell MIMO network has been investigated in \cite{Sec_mMIMO_1} and the impact of a massive MIMO relay on secrecy has been studied in \cite{Sec_mMIMO_2}.
The authors in \cite{Secure_Zhu_1} have derived two tight lower bounds for the ergodic secrecy rate considering a maximal-ratio-combining (MRC) precoder.
In order to strike a balance between complexity and performance, linear precoders based on matrix polynomials have been proposed in \cite{Secure_Zhu_2} and a phase-only zero-forcing (ZF) AN scheme has been presented in \cite{Secure_phase}.
The authors in \cite{Secure_CT} proposed a pilot-based channel training scheme for a full-duplex receiver to enhance the physical layer security.
As demonstrated in \cite{Secure_train}, AN can also be injected into the downlink training signals to prevent the eavesdropper from obtaining accurate CSI for the eavesdropping link.

Despite the promising performance gain brought by massive MIMO, it suffers from a challenging issue of high cost and power consumption due to the fact that each antenna requires a separate radio-frequency (RF) chain for signal processing.
One potential approach to reducing the required cost and power is to use digital-to-analog converters (DACs) with lower resolution for downlink transmissions \cite{DAC2}.
A number of authors have considered various direct nonlinear precoding schemes that constrain the transmit signals to match the DAC resolution.
For example, a novel precoding technique using 1-bit DACs has been presented in \cite{DAC_PSK} and a nonlinear beamforming algorithm has been proposed in \cite{DAC_pokemon}.
Also, perturbation methods minimizing the probability of error at the receivers have been studied in \cite{DAC_pert}.
An alternative simpler approach is to quantize the output of standard linear precoders, which is referred to as quantized linear precoding  \cite{DAC1}-\cite{DAC_downlink}.
Although it is generally difficult to analytically characterize the performance degradation due to nonlinear quantization, the well-known Bussgang theorem can be applied to develop an approximate linear model \cite{Bus1}, \cite{Bus2} .
This model decomposes the quantized signal into a linearly distorted version of the signal together with an uncorrelated quantization noise source \cite{J_Xu_2}.
It is noteworthy that the DAC quantization noise shares some similarities with the AN injected by the BS as both are transmitted along with the information-carrying signals and produce interference at the eavesdropper.
In other words, the DAC quantization noise can be regarded, in some sense, as a special type of AN.
Hence, it can also decrease the received signal-to-interference-and-noise ratio (SINR) at the eavesdropper, while unavoidably interfering with legitimate users at the same time.
While common sense dictates that low-resolution DAC quantization degrades system performance in conventional massive MIMO systems, it is interesting to consider the possibility that DAC quantization could enhance secrecy capacity in some scenarios.
To the best of our knowledge, only few of the existing works (e.g., \cite{Secure_0}-\cite{Sec_BDMA}, \cite{Sec_mMIMO_1}-\cite{Secure_train}) have investigated secure massive MIMO communications using low-resolution DACs.

On the other hand, although the effect of hardware impairments on secure massive MIMO systems has been analyzed in \cite{Secure_Zhu_3}, only ideal converters with infinite resolution were considered.
In this paper, we investigate secure transmission in a multiuser massive MIMO downlink network equipped with low-resolution DACs at the BS.
We assume that there exists a multi-antenna eavesdropper that intends to eavesdrop the information transmitted from the BS to multiple legitimate users.
The eavesdropper is passive in order to conceal its presence.
We assume for simplicity that perfect CSI is available at the BS since there are already a number of studies, i.e., \cite{CE_1}-\cite{CE_3}, focusing on the problem of channel estimation.
We consider two popular AN methods for injecting AN at the BS in order to prevent the unintended receiver from eavesdropping.
One method is based on AN which lies in the null-space spanned by the channels of all the desired users, while the other assumes random AN.
We also study the impact of low-resolution DACs on the achievable secrecy rate.
The main contributions of this work are summarized as follows:

1)
For the case of low-resolution DAC quantization in secure massive MIMO, we derive tight lower bounds for the secrecy rate of the system using different types of AN methods.
We observe that lower-resolution DACs provide superior secrecy performance under certain circumstances, e.g., at low SNR.
This is explained by the fact that the quantization noise degrades the eavesdropper's capacity more significantly than that of the users.
Specifically, we derive a closed-from expression for a threshold SNR $\bar{\gamma}_0$, such that if the transmit SNR $\gamma_0$ satisfies $\gamma_0<\bar{\gamma}_0$, lower-resolution DACs enhance the secrecy rate, while if $\gamma_0>\bar{\gamma}_0$, higher-resolution DACs are preferred.

2)
It is found that secure transmission with low-resolution DACs depends heavily on the power allocation factor $\phi\in(0,1]$, which denotes the proportion of power used for confidential signals, with the remainder of the power allocated for AN.
Generally, the secrecy rate first increases with $\phi$ but then subsequently decreases.
A closed-form expression for an approximate optimal $\phi^*$ is obtained. 
We observe that $\phi^*$ increases with a decreasing DAC resolution. This suggests that less power can be utilized to generate AN for DACs with a lower resolution.

3)
For the null-space AN method with the optimal $\phi^*$, we observe that low-resolution DACs lead to secrecy rate loss for all SNR values. 
On the other hand, for the random AN method, the secrecy rate with $\phi^*$ is insensitive to the DAC resolution.
This is because the DAC quantization noise behaves the same as random AN at both the intended user and eavesdropper.
As the quantization noise increases, we can maintain the same secrecy rate by reducing the power of the random AN with an increasing $\phi$.

4)
If extremely low-resolution DACs, i.e., 1-bit DACs, are employed at the BS, the advantage of null-space AN over random AN becomes marginal, while the null-space AN also suffers from a much higher computational complexity especially in massive MIMO.
In this scenario, the null-space AN method is not cost-efficient and random AN is preferred.

\begin{figure*}[tb]
\centering\includegraphics[width=0.8\textwidth,bb=50 70 650 360]{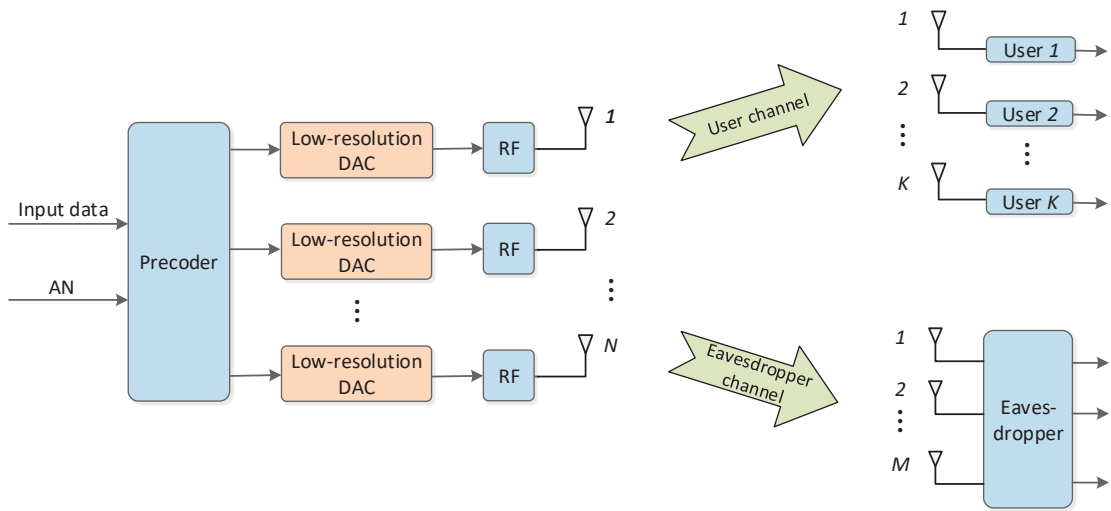}
\caption{Block diagram of the secure multiuser massive MIMO system.}
\label{block}
\end{figure*}

The rest of this paper is structured as follows.
The DAC quantization model, channel model, and two AN design methods are introduced in Section~\uppercase\expandafter{\romannumeral2}.
We derive a tight lower bound for the achievable secrecy rate in Section~\uppercase\expandafter{\romannumeral3} assuming low-resolution DACs.
Section~\uppercase\expandafter{\romannumeral4} analyzes the effect of various system parameters on secure communication.
Simulation results are presented in Section \uppercase\expandafter{\romannumeral5}, and conclusions are drawn in Section~\uppercase\expandafter{\romannumeral6}.

\emph{Notation}:
$\mathbf{A}^T$, $\mathbf{A}^*$, and $\mathbf{A}^H$ represent the transpose, conjugate, and conjugate transpose of $\mathbf{A}$, respectively.
$\mathbf{a}\sim \mathcal{CN}(\mathbf{0},\mathbf{\Sigma})$ denotes a circularly symmetric complex Gaussian vector with zero mean and covariance matrix $\mathbf{\Sigma}$.
$\textrm{tr}\{\mathbf{A}\}$ denotes the trace of $\mathbf{A}$ and $\textrm{diag}(\mathbf{A})$ is a matrix that retains only the diagonal entries of $\mathbf{A}$.
$\mathbb{E}\{\cdot\}$ is the expectation operator.
$\|\mathbf\cdot\|^2$ denotes the Euclidean norm.
$\xrightarrow{a.s.}$ denotes almost sure convergence.
$[x]^+=\mathrm{max} \{0,x\}$ chooses the maximum between $0$ and $x$.

\section{System Model}

In this section, we investigate a multiuser massive MIMO security network employing low-resolution DACs. The DAC quantization model and two AN design methods are introduced.

\subsection{Quantization Model for Low-resolution DACs}

It is in general difficult to accurately characterize the quantization error of an arbitrary low-resolution DAC. Fortunately, an equivalent linear representation has been widely adopted by using the Bussgang theorem \cite{Bus1}.
This model has been verified to be accurate enough for most DAC quantization levels in practice \cite{J_Xu_1}.
In this model, the quantized data is decomposed into two uncorrelated parts as
\begin{equation}
\label{DAC}
\mathcal{Q}_{\mathrm{DA}}(\mathbf{x})=\mathbf{F}\mathbf{x}+\mathbf{n}_{\mathrm{DA}},
\end{equation}
where $\mathcal{Q}_{\mathrm{DA}}(\cdot)$ denotes the quantization operation, $\mathbf{x}$ denotes the input data vector to the DAC, $\mathbf{F}$ represents the equivalent linear transformation matrix, and $\mathbf{n}_{\mathrm{DA}}\sim \mathcal{CN}(\mathbf{0},\mathbf{C}_{\mathrm{DA}})$ denotes the Gaussian quantization noise.
It was shown in \cite{J_Xu_1} that
\begin{equation}
\label{F}
\mathbf{F}=\sqrt{1-\rho}~\mathbf{I},
\end{equation}
and
\begin{equation}
\label{Cor_DA_0}
\mathbf{C}_{\mathrm{DA}}=\rho~\mathbb{E}\left\{\textrm{diag}\left(\mathbf{x}\mathbf{x}^H\right)\right\},
\end{equation}
where $\rho\in(0,1)$ is a distortion factor that depends on the DAC resolution $b_{\mathrm{DA}}$, which represents the number of quantized bits for the DAC.

\subsection{Secure Massive MIMO Transmission}

In the considered massive MIMO downlink network as illustrated in Fig. \ref{block}, $K$ single-antenna users are served by an $N$-antenna BS, where each transmit antenna employs a pair of low-resolution DACs for processing the in-phase and quadrature signals.
Meanwhile, a passive eavesdropper equipped with $M$ antennas strives to eavesdrop the information sent to the users.
In order to protect the confidential data from eavesdropping, the BS injects AN into the information-bearing signals.
Before transmission, the signal vector $\mathbf{s}\in \mathbb{C}^{K\times 1}$ with $\mathbb{E}\{\mathbf{s}\mathbf{s}^{H}\}=\mathbf{I}_{K}$ is precoded by a matrix $\mathbf{W}\in \mathbb{C}^{N\times K}$ with $\textrm{tr}\{\mathbf{W}\mathbf{W}^H\}=K$,
while the AN vector $\mathbf{z}\sim \mathcal{CN}(\mathbf{0},\mathbf{I}_{N-K})$ is multiplied by an AN shaping matrix $\mathbf{V}\in \mathbb{C}^{N\times (N-K)}$ with $\textrm{tr}\{\mathbf{V}\mathbf{V}^H\}=N-K$.
The weighted data vector at the BS before transmission is expressed as
\begin{align}
\mathbf{x}&=\sqrt{\frac{\phi P}{K}}\mathbf{W}\mathbf{s} + \sqrt{\frac{(1-\phi) P}{N-K}}\mathbf{V}\mathbf{z}
\triangleq\sqrt{p}\mathbf{W}\mathbf{s} + \sqrt{q}\mathbf{V}\mathbf{z}
,
\label{x}
\end{align}
where $P$ denotes the total transmit power and $\phi\in(0,1]$ is a power allocation factor.
For notational simplicity, we define
\begin{align}
p\triangleq \frac{\phi P}{K}
\label{p}
\end{align}
and
\begin{align}
q\triangleq \frac{(1-\phi) P}{N-K}.
\label{q}
\end{align}

Applying the quantization model in \eqref{DAC}, the transmit vector after DAC quantization is given by
\begin{equation}
\label{xq}
\mathbf{x}_{\mathrm{q}}=\mathcal{Q}_{\mathrm{DA}}(\mathbf{x})
=\sqrt{1-\rho}~\mathbf{x}+\mathbf{n}_{\mathrm{DA}},
\end{equation}
where $\mathbf{n}_{\mathrm{DA}}\sim \mathcal{CN}(\mathbf{0},\mathbf{C}_{\mathrm{DA}})$ represents the quantization noise which is uncorrelated with $\mathbf{x}$.
By substituting \eqref{x} into \eqref{Cor_DA_0},  the quantization noise covariance matrix $\mathbf{C}_{\mathrm{DA}}$ is obtained as
\begin{equation}
\begin{aligned}
\label{cor_DA}
\mathbf{C}_{\mathrm{DA}}=\rho\Big[ p~\textrm{diag}\left(\mathbf{W}\mathbf{W}^H\right) +q~\textrm{diag}\left(\mathbf{V}\mathbf{V}^H\right)\Big].
\end{aligned}
\end{equation}

Then, from \eqref{x} and \eqref{xq}, the received vector at the $K$ users can be expressed as
\begin{align}
\mathbf{y}&=\mathbf{H}\mathbf{x}_{\mathrm{q}}+\mathbf{n}
=\sqrt{1-\rho}\Big(\sqrt{p}\mathbf{H}\mathbf{W}\mathbf{s} + \sqrt{q}\mathbf{H}\mathbf{V}\mathbf{z}\Big)+\mathbf{H}\mathbf{n}_{\mathrm{DA}}+\mathbf{n}
,
\label{y}
\end{align}
where $\mathbf{n}\sim \mathcal{CN}(\mathbf{0},\sigma_n^2\mathbf{I}_{K})$ represents the thermal additive white Gaussian noise (AWGN) at the users,
and $\mathbf{H}\in \mathbb{C}^{K\times N}$ denotes the channel matrix between the BS and $K$ users.
In this work, we assume that long-term power control is employed to compensate for the large-scale fading of the different users.
Furthermore, the entries of $\mathbf{H}$ are modeled as independent and identically distributed (i.i.d.) complex Gaussian random variables with zero mean and unit variance.
Similarly, the received vector at the eavesdropper is
\begin{align}
\mathbf{y}_{\mathrm{e}}&=\mathbf{H}_{\mathrm{e}}\mathbf{x}_{\mathrm{q}}+\mathbf{n}_{\mathrm{e}}
=\sqrt{1-\rho}\Big(\sqrt{p}\mathbf{H}_{\mathrm{e}}\mathbf{W}\mathbf{s} + \sqrt{q}\mathbf{H}_{\mathrm{e}}\mathbf{V}\mathbf{z}\Big)+\mathbf{H}_{\mathrm{e}}\mathbf{n}_{\mathrm{DA}}+\mathbf{n}_{\mathrm{e}}
,
\label{ye}
\end{align}
where $\mathbf{n}_{\mathrm{e}}\sim \mathcal{CN}(\mathbf{0},\sigma_{\mathrm{e}}^2\mathbf{I}_{M})$ represents the thermal AWGN at the eavesdropper,
and $\mathbf{H}_{e} \in \mathbb{C}^{M \times N}$ denotes the channel matrix between the BS and the eavesdropper, whose entries are also modeled as i.i.d. complex Gaussian random variables with zero mean and unit variance.
To guarantee secure communication in the worst case, we assume that $\sigma_{\mathrm{e}}^2$ is sufficiently small at the eavesdropper and can be ignored in the sequel \cite{Sec_MIMO_AN_3}, \cite{Secure_Zhu_1}, \cite{Secure_Zhu_2}.

\subsection{AN Design Methods}
In this paper, we consider two common methods to generate the AN shaping matrix $\mathbf{V}$.
Let $\mathbf{v}_i,~\forall~i \in \{1,2,...,N-K\}$, denote the $i$th column of $\mathbf{V}$ satisfying the constraint $\|\mathbf{v}_i\|^2\!=\!1$.

\subsubsection{Null-Space Artificial Noise}
For downlink data transmission, AN is added to the transmit signals at the BS to degrade the decoding ability of the eavesdropper.
However, it can simultaneously interfere with the legitimate users as well.
In order to avoid any potential leakage of the AN to the intended users, the AN is often designed to lie in the null-space of the channel matrix $\mathbf{H}$, i.e., $\mathbf{HV}=\mathbf{0}$, assuming $\mathbf{H}$ is available at the transmitter.
However, taking low-resolution DACs into account, the AN no longer perfectly lies in the channel null-space after quantization and thus additional interference still exists.

\subsubsection{Random Artificial Noise}
For massive MIMO communication, the computational complexity of the null-space of $\mathbf{H}$ becomes prohibitively large with a large dimension $N$.
Therefore, a much simpler but effective method to design $\mathbf{V}$ was introduced in \cite{Secure_Zhu_1}.
In this method, the columns of $\mathbf{V}$ are generated as mutually independent random vectors satisfying $\|\mathbf{v}_i\|^2=1$, $\forall~i \in \{1,2,...,N-K\}$.
The random AN is inevitably leaked to the intended users but it offers much lower computational complexity compared to the null-space based AN.

Note that for both AN design methods, the columns of $\mathbf{V}$ asymptotically form an incomplete orthogonal basis with large $N$ due to the strong law of large numbers \cite{Secure_Zhu_1}.
In the following, we refer to the above two AN design methods by using superscripts, $\mathcal{N}$ and $\mathcal{R}$, respectively.

\section{Achievable Ergodic Secrecy Rate}
Given the expressions of the received signals at both the users and eavesdropper, we derive the achievable secrecy rate per user in this section, under the assumption of large numbers of antennas and users but with fixed ratios given as:
\begin{align}
\alpha \triangleq \frac{M}{N}
\label{alpha}
\end{align}
and
\begin{align}
\beta \triangleq \frac{K}{N},
\label{beta}
\end{align}
where $\beta$ denotes the user loading ratio \cite{J_Xu_2}.
To start, we first recall the following lemma from \cite[Lemma 1]{Secure_Zhu_1}.
\begin{lemma}
\label{lemma_sec_rate}
The achievable ergodic secrecy rate for the $k${\rm th} user is given by
\begin{align}
\label{secure rate}
R_{\mathrm{sec},k}=[R_k-C_k]^+,
\end{align}
where $[x]^+=\mathrm{max}\{0,x\}$, $R_k$ represents the achievable ergodic rate of the $k${\rm th} user, and $C_k$ denotes the ergodic capacity between the BS and the eavesdropper seeking to decode the information of the $k${\rm th} user.
\end{lemma}
In the following, we derive a lower bound for $R_k$ and an upper bound for $C_k$ assuming low-resolution DACs, which then provides us a lower bound for the achievable ergodic secrecy rate.

\subsection{Achievable Ergodic Rate of Each User}

From \eqref{y}, the received signal of user $k$, i.e., $y_k$, can be expressed as
\begin{align}
y_k=\sqrt{1-\rho}\left(\sqrt{p}\mathbf{h}_k^T\mathbf{W}\mathbf{s} + \sqrt{q}\mathbf{h}_k^T\mathbf{V}\mathbf{z}\right)+\mathbf{h}_k^T\mathbf{n}_{\mathrm{DA}}+n_k
,
\label{yk}
\end{align}
where $\mathbf{h}_k^T$ denotes the $k$th row of $\mathbf{H}$ and $n_k$ is the $k$th element of $\mathbf{n}$.
We also express $\mathbf{W}=[\mathbf{w}_1,\mathbf{w}_2,...,\mathbf{w}_k]$ where $\mathbf{w}_k\in \mathbb{C}^{N\times 1},~\forall k \in \{1,2,...,K\}$, is the $k$th column of $\mathbf{W}$.
Then, the signal-to-interference-quantization-and-noise ratio (SIQNR) of the $k$th user, $\gamma_k$, can be expressed as
\begin{align}
\label{SIQNR3}
&\gamma_k=
\frac{\overbrace{(1-\rho)p|\mathbf{h}_{k}^{T}\mathbf{w}_k|^2}^{S_k}}
{
\underbrace{(1\!-\!\rho)p\sum_{j\neq k}|\mathbf{h}_{k}^{T}\mathbf{w}_j|^2}_{I_k}
\!+\!\underbrace{\mathbf{h}_k^T \mathbf{C}_{\mathrm{DA}} \mathbf{h}_k^*}_{Q_k}
\!+\!\underbrace{(1\!-\!\rho)q\mathbf{h}_k^T\mathbf{VV}^H\mathbf{h}_k^*}_{A_k}
\!+\sigma_n^2}
,
\end{align}
where $S_k$ is the power of the desired signal and $I_k$ represents the power of the inter-user interference. Variables $Q_k$ and $A_k$ denote the interference power caused by DAC quantization and AN, respectively.
Then, by imposing the worst-case assumption of Gaussian distributed interference and applying Shannon's formula, a lower bound for the achievable ergodic rate of user $k$ can be evaluated as
\setcounter{equation}{15}
\begin{align}
R_k=\mathbb{E}\Big\{ \log_2\left(1+\gamma_k \right)\Big\}.
\label{Rk}
\end{align}

In order to characterize the user rate performance, we derive the asymptotic behavior of $\gamma_k$ with both AN and DAC quantization in the following lemma.

\begin{lemma}
\label{lemma_SIQNR}
Under the assumption of $N\rightarrow \infty $ with fixed $\alpha$ and $\beta$, the SIQNR of each user almost surely converges to
\begin{align}
\gamma_k^{\mathcal{N}} \xrightarrow{a.s.} \frac{(1-\rho)\left(\frac{1}{\beta}-1\right)\phi\gamma_0}{\rho\gamma_0+1} \triangleq \gamma^{\mathcal{N}},
\label{SIQNR_N}
\end{align}
for null-space AN and
\begin{align}
\gamma_k^{\mathcal{R}} \xrightarrow{a.s.} \frac{(1-\rho)\left(\frac{1}{\beta}-1\right)\phi\gamma_0}{\rho\gamma_0+(1-\rho)(1-\phi)\gamma_0+1} \triangleq \gamma^{\mathcal{R}},
\label{SIQNR_R}
\end{align}
for random AN, where $\gamma_0=\frac{P}{\sigma_n^2}$ represents the average transmit SNR.
\end{lemma}

\begin{proof}
See Appendix~A.
\end{proof}

Since convergence is preserved for continuous functions according to the Continuous Mapping Theorem \cite{Convergence}, we apply \emph{Lemma~\ref{lemma_SIQNR}} to \eqref{Rk} and thus the asymptotic achievable rates of each user for both the AN design methods are respectively obtained as
\begin{align}
\label{R_N}
R^{\mathcal{N}} = \log_2 \left(1+\frac{(1-\rho)\left(\frac{1}{\beta}-1\right)\phi\gamma_0}{\rho\gamma_0+1} \right)
\end{align}
and
\begin{align}
\label{R_R}
R^{\mathcal{R}} = \log_2 \left(1+\frac{(1-\rho)\left(\frac{1}{\beta}-1\right)\phi\gamma_0}{\rho\gamma_0+(1-\rho)(1-\phi)\gamma_0+1} \right) .
\end{align}
From \eqref{R_N} and \eqref{R_R}, it can be observed that both $R^{\mathcal{N}}$ and $R^{\mathcal{R}}$ increase with decreasing $\beta$, which implies that the achievable rate increases with more BS antennas or fewer users.
In addition, lower-resolution DACs cause higher quantization distortion with larger $\rho$, which leads to more severe user rate loss.
As $\phi$ increases, both $R^{\mathcal{N}}$ and $R^{\mathcal{R}}$ grow since more signal power is allocated to the users.
By comparing \eqref{R_N} and \eqref{R_R} with the same parameter values, it can be easily verified that $R^{\mathcal{N}}>R^{\mathcal{R}}$, as expected.
This is because random AN causes additional interference to the legitimate receivers while the more complicated null-space based AN mitigates interference leakage to the users except for the leakage due to the DAC quantization noise.
Considering extremely low-resolution DACs with $\rho\rightarrow 1$, we have $R^{\mathcal{R}}\rightarrow R^{\mathcal{N}}$ and thus random AN achieves almost the same rate performance as the null-space based AN.
Under this condition, hardly any of the AN lies in the null-space of the user's channel matrix after DAC quantization and the performance of null-space based AN tends to that of random AN.

\subsection{Ergodic Capacity of Eavesdropper}

Without loss of generality, suppose that the data of user $k$ is of interest to the eavesdropper.
In order to characterize the achievable secrecy rate, we assume the worst case that the eavesdropper has perfect knowledge of all the data channels and is able to cancel all inter-user interference before attempting to decode the message of user $k$ \cite{Sec_MIMO_AN_3}, \cite{Secure_Zhu_1}, \cite{Secure_Zhu_2}.
This assumption is reasonable because the quantization noise dominates the rate performance compared to the multiuser interference, especially for low-resolution DACs.
Using \eqref{ye} and under the assumption of large $N$ and $K$, the ergodic capacity of the eavesdropper can be evaluated as \cite{Capacity}
\begin{align}
C_k&=\mathbb{E}\Big\{\log_2\left(1+ (1-\rho)p \mathbf{w}_k^H \mathbf{H}_{\mathrm{e}}^H \mathbf{X}^{-1} \mathbf{H}_{\mathrm{e}}\mathbf{w}_k \right)\Big\}
\label{C_k0}
,
\end{align}
where $\mathbf{X}$ is defined as
\begin{align}
\mathbf{X}&\triangleq(1-\rho)q \mathbf{H}_{\mathrm{e}} \mathbf{V}\mathbf{V}^H \mathbf{H}_{\mathrm{e}}^H +  \mathbf{H}_{\mathrm{e}} \mathbf{C}_{\mathrm{DA}} \mathbf{H}_{\mathrm{e}}^H.
\label{X}
\end{align}
Since analysis of the eavesdropper's capacity in \eqref{C_k0} appears less tractable, as an alternative, we derive a tight upper bound for $C_k$, as given in the following theorem.

\begin{theorem}
\label{theorem_Ck}
For $N\rightarrow \infty$ and $\alpha+\beta<1$, an upper bound for the ergodic capacity of the eavesdropper is given by
\begin{align}
&\bar{C}\triangleq
\log_2\left(\!1\! +\!\frac{\frac{\alpha}{\beta}\phi(1-\phi+\tilde{\rho})}{\left(1\!-\!\frac{\alpha}{1\!-\!\beta}\right)(1\!-\!\phi)^2\!+\!2(1\!-\!\alpha)(1\!-\!\phi)\tilde{\rho}\!+\!(1\!-\!\alpha)\tilde{\rho}^2} \!\right)
\label{Ck_bound}
\!,
\end{align}
where $\tilde{\rho}\triangleq \frac{\rho}{1-\rho}$.
\end{theorem}

\begin{proof}
See Appendix~B.
\end{proof}

From \emph{Theorem~\ref{theorem_Ck}}, we have the following observations.

1)
The expression for the eavesdropper's capacity in \eqref{C_k0} only exists if $\mathbf{X}$ in \eqref{X} is invertible.
When $\rho\rightarrow 0$, we have $\mathbf{X}\rightarrow q \mathbf{H}_{\mathrm{e}} \mathbf{V}\mathbf{V}^H \mathbf{H}_{\mathrm{e}}^H$ since $\mathbf{C}_{\mathrm{DA}} \rightarrow \mathbf{0}$ from \eqref{cor_DA}.
In this case, $\mathbf{X}$ is invertible if $N-K > M$ since the columns of the tall matrix, $\mathbf{V}$, form an orthogonal basis for asymptotically large $N$ and the elements of $\mathbf{H}_{\mathrm{e}}$ are i.i.d. complex Gaussian distributed.
Similarly for $\rho\rightarrow 1$, $\mathbf{X}\rightarrow \mathbf{H}_{\mathrm{e}}\left[ p~\textrm{diag}(\mathbf{W}\mathbf{W}^H) +q~\textrm{diag}(\mathbf{V}\mathbf{V}^H)\right]\mathbf{H}_{\mathrm{e}}^H$ is invertible if $N>M$.
Combining the above two conditions, we see that $\mathbf{X}$ is invertible when $N-K > M$ regardless of the value of $\rho\in(0,1)$.
This results in the same constraint, i.e., $\alpha+\beta<1$, as in \emph{Theorem~\ref{theorem_Ck}}, and is a common condition for massive MIMO systems with a large $N$.

2)
From \eqref{Ck_bound}, it is obvious that $\bar{C}$ is monotonically increasing with $\alpha$.
This implies that the BS can reduce the amount of private information leaked to the eavesdropper by deploying more transmit antennas,
while the eavesdropper can improve its wiretapping capability by employing more receive antennas.

3)
Given $\alpha$, $\rho$, and $\phi$, the effect of $\beta$ on $\bar{C}$ is generally not monotonic.
By characterizing the derivative of $\bar{C}$ with respect to (w.r.t.) $\beta$, we find that $\bar{C}$ decreases for $\beta\in(0,\bar{\beta})$, while it increases when $\beta\in(\bar{\beta},1-\alpha)$, where
\begin{align}
\label{beta_bar}
\bar{\beta}\triangleq 1-\sqrt{\frac{\alpha(1-\phi)^2}{(1-\alpha)\left[(1-\phi)+\tilde{\rho}\right]^2+\alpha(1-\phi)^2}}.
\end{align}
This can be explained as follows.
When $\beta$ is small, the transmit power allocated to each user decreases significantly with increasing $\beta$ and thus the eavesdropper's capacity decreases accordingly.
As $\beta$ continues increasing, the impact of the reduced power per user becomes less significant.
When $\beta$ approaches $1-\alpha$, $\mathbf{X}$ becomes ill-conditioned and the eavesdropper's capacity improves.
In addition, it is noted that $\bar{\beta}$ can be larger than $1-\alpha$ for large values of $\rho$ and $\phi$.
Under this condition, $\bar{C}$ decreases monotonically for $\beta\in(0,1-\alpha)$.

4)
The parameter $\tilde{\rho}\in(0,\infty)$ represents the influence of the low-resolution DACs on the capacity of the eavesdropper.
By characterizing the derivative of $\bar{C}$ w.r.t. $\tilde{\rho}$, we find that $\frac{\partial \bar{C} }{\partial \tilde{\rho}}<0,~\forall \tilde{\rho}$.
It implies that $\bar{C}$ decreases with $\tilde{\rho}$, and hence with $\rho$.
Since $\rho$ increases with decreasing DAC resolution $b_{\mathrm{DA}}$, a smaller $b_{\mathrm{DA}}$ leads to a lower $\bar{C}$ due to the increasing power of the quantization noise.
This implies that the utilization of low-resolution DACs makes some contribution to protecting the legitimate users from eavesdropping, although it concurrently decreases the achievable user rate.

5)
It is found that $\bar{C}$ increases with $\phi$, i.e., $\frac{\partial \bar{C} }{\partial \phi}>0$, as the eavesdropper's capacity increases with decreasing AN power.
Assuming that there is no AN, i.e., $\phi=1$, $\bar{C}$ in \eqref{Ck_bound} achieves the maximum which is given by
\begin{align}
\bar{C}&=\log_2\left[1+ \frac{\alpha}{(1-\alpha)\beta\tilde{\rho}} \right]
\label{Ck_bound_2}
.
\end{align}
Note that $\bar{C}$ does not grow without an upper bound even if AN is not present due to the low-resolution DAC quantization.
To a certain extent, the quantization noise acts as a type of AN which helps to degrade the eavesdropper's capacity by producing unavoidable interference.
In this case, $\bar{C}$ becomes a monotonically decreasing function w.r.t. $\beta\in(0,1-\alpha)$ because $\bar{\beta}=1>1-\alpha$ by substituting $\phi=1$ into \eqref{beta_bar}.

\subsection{Lower Bound for the Achievable Secrecy Rate}
Applying \emph{Lemma~\ref{lemma_sec_rate}} and using \eqref{R_N}, \eqref{R_R}, and \eqref{Ck_bound}, a lower bound for the achievable secrecy rate of each user is obtained as follows
\begin{align}
\label{Sec_rate}
\underline{R}_\mathrm{sec}^\Psi=\left[R^\Psi-\bar{C}\right]^+,
\end{align}
where $\Psi\in\{\mathcal{N},\mathcal{R}\}$.
Using the results derived above, expressions for $\underline{R}_\mathrm{sec}^\mathcal{N}$ and $\underline{R}_\mathrm{sec}^\mathcal{R}$ are respectively obtained as
\begin{align}
\underline{R}_\mathrm{sec}^\mathcal{N} =
\left[\log_2\left(1+\frac{(1-\rho)\left(\frac{1}{\beta}-1\right)\phi\gamma_0}{\rho\gamma_0+1}\right)
-\log_2\left(1+\frac{\alpha\phi\left(\frac{1}{\beta}-1\right)\mu}{(\nu+\alpha\beta)\mu^2-\zeta}
\right)\right]^+
\label{Sec_R_N}
,
\end{align}
and
\begin{align}
\underline{R}_\mathrm{sec}^\mathcal{R} =
\left[\log_2\left(1+\frac{(1-\rho)\left(\frac{1}{\beta}-1\right)\phi\gamma_0}{\rho\gamma_0+(1-\rho)(1-\phi)\gamma_0+1}\right)
-\log_2\left(1+\frac{\alpha\phi\left(\frac{1}{\beta}-1\right)\mu}{(\nu+\alpha\beta)\mu^2-\zeta}
\right)\right]^+
\label{Sec_R_R}
\!,
\end{align}
where we define $\nu\triangleq 1-\alpha-\beta$, $\mu\triangleq 1-\phi+\tilde{\rho}$, and $\zeta=\alpha\beta(1-\phi)^2$ for notational simplicity.
These closed-form expressions allow us to gain insight into the impact of the various system parameters, as detailed in the next sections.

\section{Secrecy Rate Analysis}

In this section, we analyze the impact of various parameters, including $\alpha$, $\beta$, $\rho$, and $\phi$, on the secrecy rate in massive MIMO systems using low-resolution DACs.

\subsection{Impact of Antenna and User Loading Ratios}
We first analyze the impact of the antenna ratio $\alpha$ defined in \eqref{alpha}.
In \eqref{Sec_rate}, $\bar{C}$ increases monotonically with $\alpha$ as indicated before while $R^\Psi$ is independent of $\alpha$.
As a consequence, $\underline{R}_\mathrm{sec}^\Psi$ is monotonically decreasing w.r.t. $\alpha$.
Thus, a threshold value, $\bar{\alpha}$, may exist such that no positive secrecy rate can be achieved when $\alpha>\bar{\alpha}$, regardless of the values of other parameters.
In other words, secure transmission cannot be achieved if the eavesdropper possesses enough antennas.

Since AN is injected to enhance the secrecy rate, we consider the special case that almost all the power is allocated to generate AN, i.e., $\phi\rightarrow 0$.
By setting $\underline{R}_\mathrm{sec}^\Psi=0$ in \eqref{Sec_R_N} and \eqref{Sec_R_R}, $\bar{\alpha}$ is obtained as
\setcounter{equation}{28}
\begin{align}
\bar{\alpha}^{\mathcal{N}}
=\frac{(1-\beta)\gamma_0}{(\rho+1)\gamma_0+1-\beta\gamma_0\rho(2-\rho)}
\label{alpha_S_N}
\end{align}
and
\begin{align}
\bar{\alpha}^{\mathcal{R}}
=\frac{(1-\beta)\gamma_0}{2\gamma_0+1-\beta\gamma_0\rho(2-\rho)}.
\label{alpha_S_R}
\end{align}
Since $\rho\in(0,1)$, we have $\bar{\alpha}^{\mathcal{N}}>\bar{\alpha}^{\mathcal{R}}$, which implies that the null-space based AN can tolerate a larger number of eavesdropper antennas than the random AN at the expense of higher computational complexity and the need for CSI.
Interestingly, it can be observed that $\bar{\alpha}^{\mathcal{N}}\rightarrow\bar{\alpha}^{\mathcal{R}}$ when $\rho\rightarrow 1$.
This is because the null-space based AN tends to be randomly distributed in the signal space after low-resolution DAC quantization.
Note that both $\bar{\alpha}^{\mathcal{N}}$ and $\bar{\alpha}^{\mathcal{R}}$ decrease with $\beta$.
Next, we focus on the extreme condition when $\beta$ reduces to near $0$:
\begin{align}
\lim_{\beta\rightarrow 0} \bar{\alpha}^{\mathcal{N}}
=\frac{\gamma_0}{(\rho+1)\gamma_0+1}
\label{alpha_S_N_2}
\end{align}
and
\begin{align}
\lim_{\beta\rightarrow 0} \bar{\alpha}^{\mathcal{R}}
=\frac{\gamma_0}{2\gamma_0+1}.
\label{alpha_S_R_2}
\end{align}
Under this circumstance, $\lim\limits_{\beta\rightarrow 0}\bar{\alpha}^{\mathcal{R}}$ is independent of $\rho$ because the DAC quantization does not statistically change the randomness of the random AN.
By increasing $\gamma_0$, both $\lim\limits_{\beta\rightarrow 0}\bar{\alpha}^{\mathcal{N}}$ and $\lim\limits_{\beta\rightarrow 0}\bar{\alpha}^{\mathcal{R}}$ grow accordingly, thus improving the robustness for both AN design methods.
In all cases, however, the two thresholds are ultimately bounded above by $\lim\limits_{\beta\rightarrow 0}\bar{\alpha}^{\mathcal{N}}<\frac{1}{\rho+1}$ and $\lim\limits_{\beta\rightarrow 0}\bar{\alpha}^{\mathcal{R}}<\frac{1}{2}$, respectively.

One can also study the impact of user loading ratio $\beta$ defined in \eqref{beta}.
We take the derivative of $\underline{R}_\mathrm{sec}^\Psi$ w.r.t. $\beta$ and obtain that $\frac{\partial \underline{R}_\mathrm{sec}^\Psi}{\partial \beta}<0$.
Hence, by combining the observations from \eqref{Ck_bound}, it is reasonable to expect that the secrecy rate will be enhanced with a smaller $\beta$.
Furthermore, adding more antennas at the BS can offer a larger beamforming gain, or alternatively, a smaller number of users leads to higher per-user transmit power.

\subsection{Impact of DAC Distortion Parameter}
Since both $R^\Psi$ and $\bar{C}$ decrease with increasing $\rho$ due to the low-resolution DAC quantization, the impact of $\rho$ on the secrecy rate, $\underline{R}_\mathrm{sec}^\Psi$, is unclear.
According to \eqref{Sec_rate} and assuming a positive secrecy rate, we have
\begin{align}
\frac{\partial \underline{R}_\mathrm{sec}^{\Psi} }{\partial \rho}=\frac{\partial R^{\Psi} }{\partial \rho}-\frac{\partial \bar{C} }{\partial \rho}.
\end{align}
On one hand, $\frac{\partial \bar{C} }{\partial \rho}<0$ is independent of $\gamma_0$ since we assume a near-zero thermal noise power at the eavesdropper.
On the other hand, $\frac{\partial R^{\Psi} }{\partial \rho}<0$ and decreases with large $\gamma_0$ because the quantization noise dominates the thermal noise at high SNRs.
Thus, we conclude that there exists a $\bar{\gamma}_0^{\Psi}\in(0,\infty)$ which guarantees that $\frac{\partial \underline{R}_\mathrm{sec}^{\Psi} }{\partial \rho}>0$ for $\gamma_0\in(0,\bar{\gamma}_0^{\Psi})$ and $\frac{\partial \underline{R}_\mathrm{sec}^{\Psi} }{\partial \rho}<0$ for $\gamma_0\in(\bar{\gamma}_0^{\Psi},\infty)$.
Interestingly, lower-resolution DACs can achieve higher secrecy rate at low SNR, because the eavesdropper's capacity $\bar{C}$ decreases faster than $R^{\Psi}$ does with an increasing $\rho$.
On the other hand, at high SNR, higher-resolution DACs are advantageous compared to those with lower-resolution.

For the null-space AN method, the expression for $\frac{\partial \underline{R}_\mathrm{sec}^{\Psi} }{\partial \rho}$ is given below:
\begin{align}
\frac{\partial \underline{R}_\mathrm{sec}^{\mathcal{N}} }{\partial \rho}
 = &-\frac{\left(\frac{1}{\beta}-1\right)\phi(\gamma_0+1)\gamma_0}{\ln 2~(\rho\gamma_0\!+\!1)\left[\rho\gamma_0+(1-\rho)\left(\frac{1}{\beta}-1\right)\phi\gamma_0+1\right]}
\nonumber
\\
& + \frac{\alpha\phi\left(\frac{1}{\beta}-1\right)\left[(\nu+\alpha\beta)\mu^2+\zeta\right]}
 {\ln 2(\!1\!-\!\rho\!)^2\!\!\left[(\nu\!+\!\alpha\beta)\mu^2\!-\!\zeta\right]\!\!\left[\!(\nu\!+\!\alpha\beta)\mu^2\!-\!\zeta+\alpha\phi\left(\!\frac{1}{\beta}\!-\!1\!\right)\mu\!\right]}
 \label{d_Rsec_2}
 \\
\triangleq &~\frac{1}{\ln 2}\frac{a^{\mathcal{N}}\gamma_0^2+b^{\mathcal{N}}\gamma_0+c^{\mathcal{N}}}{d^{\mathcal{N}}}
 \label{d_Rsec}
 ,
\end{align}
where \eqref{d_Rsec_2} utilizes $\frac{\partial \mu }{\partial \rho}=\frac{\partial \tilde{\rho} }{\partial \rho}=\frac{1}{(1-\rho)^2}$.
Obviously, the sign of $\frac{\partial \underline{R}_\mathrm{sec}^{\mathcal{N}} }{\partial \rho}$ depends on the values of the parameters $\alpha$, $\beta$, $\rho$, and $\phi$.
We have focused on the impact of $\gamma_0$ and regard the derivative as a quadratic equation w.r.t $\gamma_0$ as in \eqref{d_Rsec}.
In general, we have that $d^{\mathcal{N}}>0$, $a^{\mathcal{N}}<0$, and $c^{\mathcal{N}}>0$. This implies that a solution for $\bar{\gamma}_0^\mathcal{N}$ exists by forcing \eqref{d_Rsec} to zero.
Solving the quadratic yields
\begin{align}
\label{gamma_bar_N}
\bar{\gamma}_0^{\mathcal{N}}=\frac{-b^{\mathcal{N}}-\sqrt{{b^{\mathcal{N}}}^2-4a^{\mathcal{N}}c^{\mathcal{N}}}}{2a^{\mathcal{N}}}.
\end{align}
If $\gamma_0<\bar{\gamma}_0^{\mathcal{N}}$, lower-resolution DACs can be used to enhance the secrecy rate since quantization noise degrades the eavesdropper's capacity more pronouncedly than the user rate.
While for $\gamma_0>\bar{\gamma}_0^{\mathcal{N}}$, the infinite-resolution DACs achieve the best performance.
Since the expressions for $a^{\mathcal{N}}$, $b^{\mathcal{N}}$, $c^{\mathcal{N}}$, and $d^{\mathcal{N}}$ are generally complicated, we consider a special case with $\rho\rightarrow 0$, which means that ideal DACs with infinite resolution are assumed.
Under this condition, the related parameters are obtained as
\begin{align}
\label{a_N}
a^{\mathcal{N}}= -\nu(1-\phi)\phi \left[ \nu(1-\phi)+\alpha\phi\left(\frac{1}{\beta}-1\right)\right],~~~~~~~~
\end{align}
\begin{align}
\label{b_N}
b^{\mathcal{N}}= 2\alpha^2\phi^2(1\!-\!\beta)+\alpha\phi^3\nu\left(\frac{1}{\beta}\!-\!1\right)\!-\!\nu^2(1\!-\!\phi)^2\phi
,~~~
\end{align}
\begin{align}
\label{c_N}
c^{\mathcal{N}}= \alpha \phi (\nu+2\alpha\beta),~~~~~~~~~~~~~~~~~~~~~~~~~~~~~~~~~~~~~~~~
\end{align}
\begin{align}
d^{\mathcal{N}}\!=\!\ln2~\nu(1\!-\!\phi)\left[\! \frac{\nu(1\!-\!\phi)\beta}{1-\beta}\!+\!\alpha\phi\right]\left[\left(\frac{1}{\beta}\!-\!1\right)\phi\gamma_0\!+\!1 \!\right]
\label{d_N}
\!.
\end{align}
By substituting \eqref{a_N}-\eqref{d_N} into \eqref{gamma_bar_N}, the threshold $\bar{\gamma}_0^{\mathcal{N}}$ for $\rho\rightarrow 0$ is obtained.
Although the threshold relies on $\rho$ in general, the obtained $\bar{\gamma}_0^{\mathcal{N}}$ can approximately be applied to all values of $\rho\in(0,1)$, which is verified by the simulation results in Section~\uppercase\expandafter{\romannumeral5}.

For the random AN design method, similar manipulations can be conducted and the threshold SNR $\bar{\gamma}_0^{\mathcal{R}}$ is obtained as follows
\begin{align}
\label{gamma_bar_R}
\bar{\gamma}_0^{\mathcal{R}}=\frac{-b^{\mathcal{R}}-\sqrt{{b^{\mathcal{R}}}^2-4a^{\mathcal{R}}c^{\mathcal{R}}}}{2a^{\mathcal{R}}},
\end{align}
where
\begin{align}
\label{a_R}
a^{\mathcal{R}}=& \!-\!\nu(1-\phi)\phi \left[ \nu(1-\phi)+\alpha\phi\left(\frac{1}{\beta}-1\right)\right]
\!+\!\alpha\phi(\!1\!-\!\phi)(\nu+2\alpha\beta)\left[\left(\frac{1}{\beta}\!-\!1\right)\phi\!+\!1\!-\!\phi\right]
,
\end{align}
\begin{align}
\label{b_R}
b^{\mathcal{R}}=& 2\alpha^2\phi^2(1-\beta)+\alpha\phi^3\nu\left(\frac{1}{\beta}-1\right)-\nu^2(1-\phi)^2\phi
+2\alpha\phi(1-\phi)(\nu+2\alpha\beta)
,~~~~~~~~~~~~
\end{align}
\begin{align}
\label{c_R}
c^{\mathcal{R}}= \alpha \phi (\nu+2\alpha\beta)
,~~~~~~~~~~~~~~~~~~~~~~~~~~~~~~~~~~~~~~~~~~~~~~~~~~~~~~~~~~~~~~~~~~~~~~~~~~~~~~~~~
\end{align}
\begin{align}
d^{\mathcal{R}}=&\ln2\!~\!\nu(1\!-\!\phi)\left[ \frac{\nu(1-\phi)\beta}{1-\beta}+\alpha\phi\right]\left[(1-\phi)\gamma_0\!+\!1\right]
\left[\left(\frac{1}{\beta}-1\right)\phi\gamma_0+(1-\phi)\gamma_0+1\right]
\label{d_R}
.~~~
\end{align}
Similarly, $a^{\mathcal{R}}$, $b^{\mathcal{R}}$, $c^{\mathcal{R}}$, and $d^{\mathcal{R}}$ are obtained under the assumption of $\rho\rightarrow 0$ and $\bar{\gamma}_0^{\mathcal{R}}$ is also insensitive to the value of $\rho$.

\subsection{Impact of the Power Allocation Factor}
The above analysis was conducted assuming a fixed $\phi$. Now, we investigate the effect of this power allocation factor on the secrecy rate.
Since $\frac{\partial R^\Psi}{\partial \phi}>0$ and $\frac{\partial \bar{C}}{\partial \phi}>0$ as indicated above, the sign of $\frac{\partial \underline{R}_\mathrm{sec}^\Psi}{\partial \phi}=\frac{\partial R^\Psi}{\partial \phi}-\frac{\partial \bar{C}}{\partial \phi}$ cannot be immediately determined.

Take the secrecy rate in \eqref{Sec_R_N} with the null-space AN for instance.
The derivative of $\underline{R}_\mathrm{sec}^\mathcal{N}$ w.r.t. $\phi$ is calculated as
\begin{align}
&\frac{\partial \underline{R}_\mathrm{sec}^\mathcal{N}}{\partial \phi}\!=\!
\frac{(1-\rho)(\frac{1}{\beta}-1)\gamma_0}{\ln2\left[\rho\gamma_0\!+\!1\!+\!(\!1\!-\!\rho)(\frac{1}{\beta}\!-\!1\!)\gamma_0\phi\right]}
\!-\!\frac{\alpha\!\left(\!\frac{1}{\beta}\!-\!1\!\right)\!\!\left[\!\frac{(\nu\!+\!\alpha\beta)\mu^2}{1-\rho}\!-\!2\alpha\beta\mu\phi(\!1\!-\!\phi)\!-\!\alpha\beta(\!1\!-\!\phi)^2(\mu\!-\!\phi)\!\right]}
{\ln2~[(\nu\!+\!\alpha\beta)\mu^2\!-\!\zeta]\!\left[\!(\nu\!+\!\alpha\beta)\mu^2\!-\!\zeta+\alpha\phi\left(\!\frac{1}{\beta}\!-\!1\!\right)\mu\right]},
\label{d_Rs_phi}
\end{align}
where we use the fact that $\frac{\partial \mu}{\partial \phi}=-1$.
For small $\phi$ we have $\frac{\partial \underline{R}_\mathrm{sec}^\mathcal{N}}{\partial \phi}>0$ while for large $\phi$ we have $\frac{\partial \underline{R}_\mathrm{sec}^\mathcal{N}}{\partial \phi}<0$.
Thus, there exists an optimal $\phi$, i.e., $\phi^*$, that achieves the highest secrecy rate.
By forcing $\frac{\partial \underline{R}_\mathrm{sec}^\mathcal{N}}{\partial \phi}=0$, the optimal $\phi^*$ is directly obtained.
Since the expression in \eqref{d_Rs_phi} is generally intractable, we resort to the numerical bisection method to determine $\phi^*$.
In addition, we derive a closed-form expression for an approximate $\phi^*$ in the following.
We assume that $\alpha\beta\ll 1$, which generally holds in massive MIMO networks with large antenna arrays at the BS.
Then, $\frac{\partial \underline{R}_\mathrm{sec}^\mathcal{N}}{\partial \phi}$ in \eqref{d_Rs_phi} approximately becomes
\begin{align}
\frac{\partial \underline{R}_\mathrm{sec}^\mathcal{N}}{\partial \phi}=&
\frac{(1-\rho)(\frac{1}{\beta}-1)\gamma_0}{\ln~2\left[\rho\gamma_0+1+(1-\rho)(\frac{1}{\beta}-1)\gamma_0\phi\right]}
-\frac{\alpha\left(\!\frac{1}{\beta}\!-\!1\!\right)\frac{\nu\mu^2}{1-\rho}}
{\ln2~\nu\mu^2\left[\nu\mu^2+\alpha\phi\left(\!\frac{1}{\beta}\!-\!1\!\right)\mu\right]}.
\label{d_Rs_phi_2}
\end{align}
Setting $\frac{\partial \underline{R}_\mathrm{sec}^\mathcal{N}}{\partial \phi}=0$, the optimal $\phi^*$ is obtained as
\begin{align}
\phi^{\mathcal{N}*}=\frac{\nu-\sqrt{\nu^2+\left(\alpha\rho+\frac{\alpha}{\gamma_0}-\nu\right)\left(1-\beta-\frac{\alpha}{\beta}\right)}}{(1-\rho)\left(1-\beta-\frac{\alpha}{\beta}\right)}.
\label{phi_star_N}
\end{align}

For random AN, a similar analysis can be conducted.
Under the same assumption $\alpha\beta\ll 1$, the optimal $\phi^{\mathcal{R}*}$ is given by
\begin{align}
&\phi^{\mathcal{R}*}=
\frac{(1\!+\!\gamma_0)(\nu-\alpha)\!-\!\sqrt{\alpha(\!1\!+\!\gamma_0\!)\!\left[\!\left(\!\frac{1}{\beta}\!+\!1\!\right)\!(\!\nu\!-\!\alpha\!)\!+\!\frac{1}{\gamma_0}\!\left(\!1\!-\!\beta\!-\!\frac{\alpha}{\beta}\!\right)\!\right]}}
{(1-\rho)\left[1-\beta-\frac{\alpha}{\beta}+\gamma_0(\nu-\alpha)\right]}.
\label{phi_star_R}
\end{align}

Due to the constraint that $\phi\in(0,1]$, we set $\phi^*=1$ if the obtained $\phi^*$ in \eqref{phi_star_N} and \eqref{phi_star_R} is larger than $1$. Under this condition, the secrecy rate increases monotonically with $\phi\in(0,1]$.
In Section~\uppercase\expandafter{\romannumeral5}, we will show that both $\phi^{\mathcal{N}*}$ in \eqref{phi_star_N} and $\phi^{\mathcal{R}*}$ in \eqref{phi_star_R} are accurate for various combinations of system parameters.

\section{Simulation Results}
In this section, we verify the tightness of the derived bounds and the obtained insights via numerical simulation.
We use the typical values for the distortion parameter $\rho$ in \cite{rho} for each DAC using $b_{\mathrm{DA}}$ bits for quantization.
For perfect DACs with $b_{\mathrm{DA}}\rightarrow\infty$, we set $\rho\rightarrow 0$.

\subsection{Ergodic Capacity of Eavesdropper}

\begin{figure}[tb]
\centering\includegraphics[width=0.5\textwidth,bb=20 200 570 610]{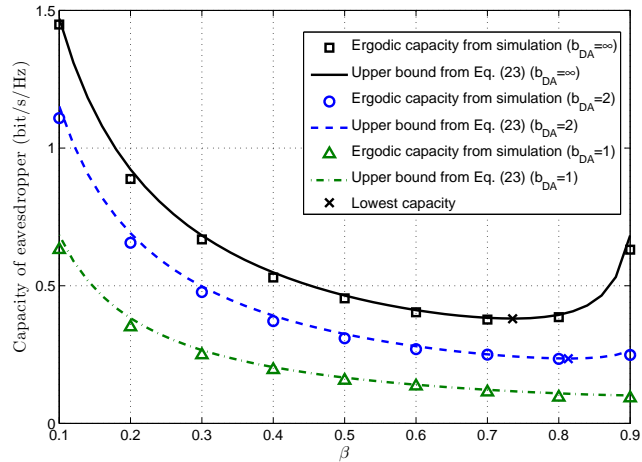}
\caption{Eavesdropper's capacity and the corresponding upper bounds versus $\beta$ ($N=100$, $M=7$, and $\phi=0.7$).}
\label{C_beta}
\end{figure}

\begin{figure}[tb]
\centering\includegraphics[width=0.5\textwidth,bb=20 200 570 610]{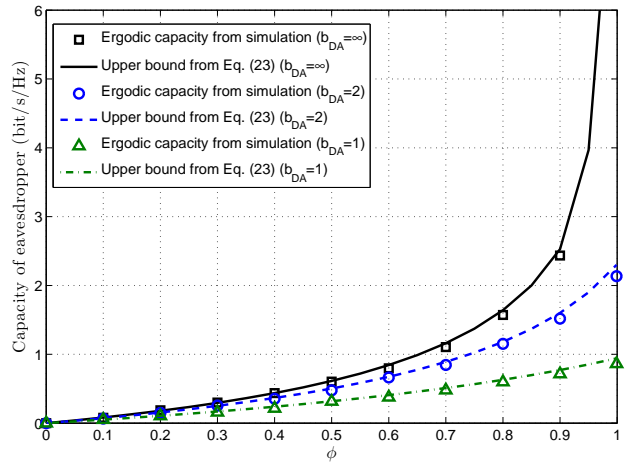}
\caption{Eavesdropper's capacity and our derived upper bound versus power allocation factor $\phi$ ($N=100$, $K=10$, and $M=5$).}
\label{C_alpha}
\end{figure}

We first study the tightness of the derived upper bound for the eavesdropper's capacity.
Fig. \ref{C_beta} compares the eavesdropper's capacity in \eqref{C_k0} and the upper bound in \eqref{Ck_bound}, for DAC resolutions $b_{\mathrm{DA}}=1, 2$, and $\infty$.
In general, our derived upper bound is tight for $\beta$ ranging from $0.1$ to $0.9$.
Clearly, the low-resolution DACs result in a capacity loss due to the interference caused by the quantization noise.
As accurately predicted by our analysis in \eqref{beta_bar}, the eavesdropper achieves the lowest capacity for $\bar{\beta}=0.7354$ and $\bar{\beta}=0.8133$ with $b_{\mathrm{DA}}=\infty$ and $b_{\mathrm{DA}}=2$, respectively.
These two points are denoted by markers $\times$ in the figure.
For $b_{\mathrm{DA}}=1$, we have $\bar{\beta}=0.9059$ according to \eqref{beta_bar} and thus $\bar{C}$ decreases monotonically for $\beta\in(0.1,0.9)$.

Fig. \ref{C_alpha} shows the capacity of the eavesdropper for $\phi$ ranging from $0$ to $1$.
Obviously, $\bar{C}$ increases monotonically with $\phi$. The lower the AN power, the higher the eavesdropper's capacity will be due to the power reduction in the interference.
In addition, we see that low-resolution DACs help to degrade the channel quality of the eavesdropper regardless of the value of $\phi$.
Assuming the eavesdropper is able to perfectly cancel the inter-user interference and the thermal noise is negligibly small, the capacity approaches infinity with $\phi\rightarrow 1$ and $b_{\mathrm{DA}}=\infty$ since there is no remaining interference.
Thus, AN is necessary for conventional secure communication when perfect DACs are available.
However, this is not the case for low-resolution DACs since the quantization noise can protect the confidential information from eavesdropping.
Under the assumption of $\phi\rightarrow 1$, the capacity converges to $2.2985$ and $0.9407$, instead of infinity, for $b_{\mathrm{DA}}=2$ and $b_{\mathrm{DA}}=1$, respectively.

\subsection{Achievable Ergodic Secrecy Rate}

\begin{figure}[tb]
\centering\includegraphics[width=0.5\textwidth,bb=20 200 570 610]{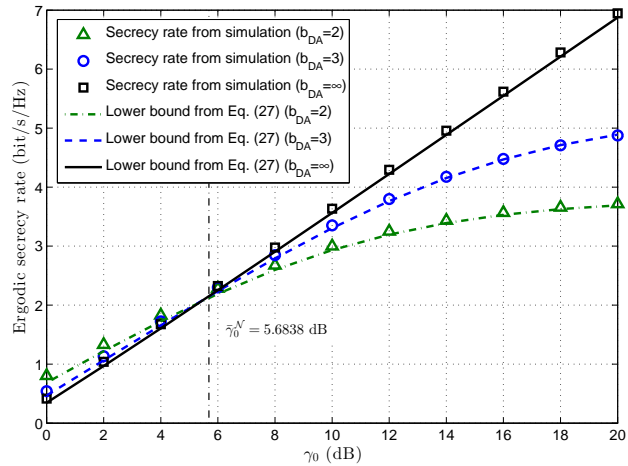}
\caption{Ergodic secrecy rate and lower bound versus SNR with null-space AN ($N=128$, $K=8$, $M=16$, and $\phi=0.8$).}
\label{R_sec_N}
\end{figure}

In the following, we verify the accuracy of the derived lower bound for the achievable secrecy rate.
Fig.~\ref{R_sec_N} shows the ergodic secrecy rate and its lower bound in \eqref{Sec_R_N} with the null-space AN method.
The dotted markers correspond to the simulation results while the solid lines correspond to the lower bound.
We observe that the derived bound is tight for $\gamma_0$ ranging from 0~dB to 20~dB.
With infinite-resolution DACs, the secrecy rate increases proportionally with $\gamma_0$ while the low-resolution DAC quantization causes significant rate loss at high SNR.
From \eqref{gamma_bar_N}, the SNR threshold is computed as $\bar{\gamma}_0^\mathcal{N}=5.6838$~dB with $\rho\rightarrow 0$, i.e., $b_{\mathrm{DA}}\rightarrow \infty$.
When $\gamma_0<\bar{\gamma}_0^\mathcal{N}$, lower-resolution DACs can provide higher secrecy rate since the achievable rate of each user decreases more slowly than the eavesdropper's capacity as the DAC resolution decreases.
At low SNR, thermal noise dominates at the users and the DAC quantization affects the eavesdropper's capacity more pronouncedly.
On the other hand, when $\gamma_0>\bar{\gamma}_0^\mathcal{N}$, infinite-resolution DACs achieve the highest secrecy rate.
In addition, we observe that the obtained $\bar{\gamma}_0^\mathcal{N}$ can also be applied to low-resolution DACs with $b_{\mathrm{DA}}=3$ and $b_{\mathrm{DA}}=2$ as indicated before, although technically $\bar{\gamma}_0^\mathcal{N}$ depends on the quantization distortion parameter $\rho$.
This makes the DAC resolution allocation much simpler and appealing in practice since only one $\bar{\gamma}_0^\mathcal{N}$ needs to be calculated.

\begin{figure}[tb]
\centering\includegraphics[width=0.5\textwidth,bb=20 200 570 610]{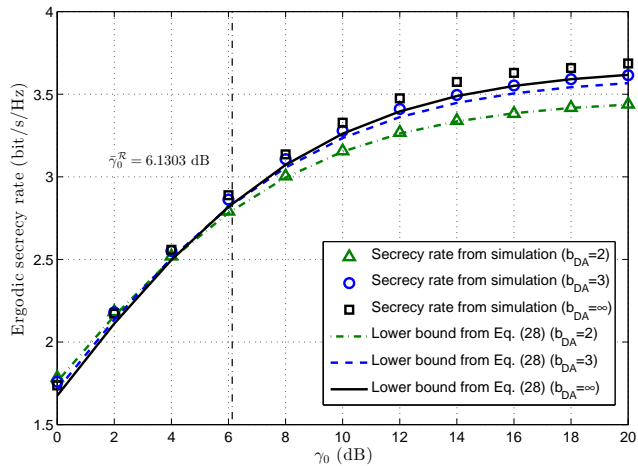}
\caption{Ergodic secrecy rate and lower bound versus SNR with random AN ($N=128$, $K=8$, $M=6$, and $\phi=0.7$).}
\label{R_sec_R}
\end{figure}

Fig.~\ref{R_sec_R} illustrates the ergodic secrecy rate and the derived lower bound in \eqref{Sec_R_R} with random AN.
The secrecy rate increases with SNR but saturates eventually at high SNR, even if infinite-resolution DACs are adopted.
This is because the random AN degrades the achievable rate of the legitimate users while the null-space AN only causes interference to the eavesdropper.
From \eqref{gamma_bar_R}, $\bar{\gamma}_0^\mathcal{R}$ is calculated as $6.1303$~dB.
In order to enhance the secrecy rate, increasing the DAC resolution is recommended if $\gamma_0>6.1303$~dB but not if $\gamma_0\leq6.1303$~dB.
In both Fig.~\ref{R_sec_N} and Fig.~\ref{R_sec_R}, a fixed $\phi$ is assumed since it is in general difficult to optimize $\phi$ analytically.
In order to alleviate the performance degradation of fixed power allocation, we present an approximate $\phi^*$ in \eqref{phi_star_N} and \eqref{phi_star_R} for the null-space AN and the random AN methods, respectively.
Corresponding simulations are illustrated in Fig.~\ref{R_opt_phi_N} and Fig.~\ref{R_opt_phi_R} in the next subsection.

\begin{figure}[tb]
\centering\includegraphics[width=0.5\textwidth,bb=20 200 570 610]{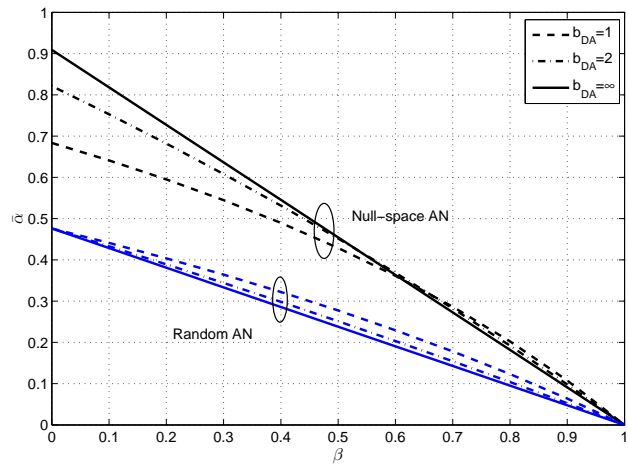}
\caption{Threshold ratio $\bar{\alpha}$ for positive secrecy rate versus $\beta$ ($\gamma_0=10$~dB).}
\label{alpha_bar}
\end{figure}

Fig. \ref{alpha_bar} depicts $\bar{\alpha}$ in \eqref{alpha_S_N} and \eqref{alpha_S_R} for null-space and random AN, respectively.
As indicated in Section~\uppercase\expandafter{\romannumeral4}.A, a positive secrecy rate can be achieved only if $\alpha<\bar{\alpha}$.
It is observed from Fig. \ref{alpha_bar} that $\bar{\alpha}$ decreases monotonically with $\beta$.
Given a fixed $N$, the transmit power of each user decreases with an increasing number of users $K$, i.e., increasing $\beta$, and thus fewer antennas are required at the eavesdropper to decode the information.
Note that even with $\beta\rightarrow 0$, a threshold $\bar{\alpha}<1$ exists, which implies that the eavesdropper is still able to successfully wiretap as long as it employs enough antennas.
Comparing null-space and random AN, we find that $\bar{\alpha}^\mathcal{N}>\bar{\alpha}^\mathcal{R}$ for $\gamma_0=10$~dB.
This implies that higher hardware cost is required at the eavesdropper to resist null-space AN than random AN.

\subsection{Optimal Power Allocation}

\begin{figure}[tb]
\centering\includegraphics[width=0.5\textwidth,bb=20 200 570 610]{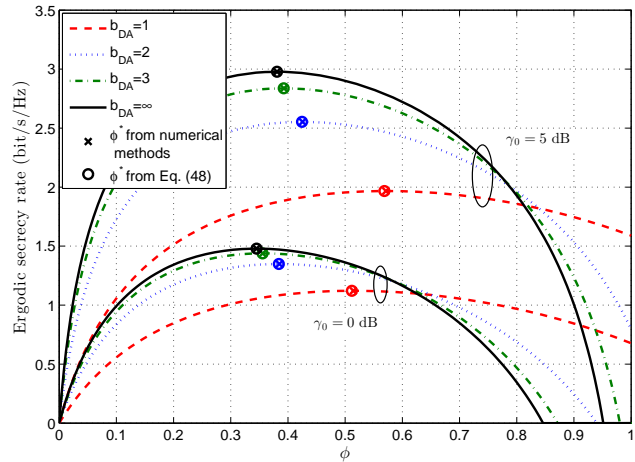}
\caption{Ergodic secrecy rate versus $\phi$ with null-space AN ($N=128$, $K=8$, and $M=16$).}
\label{R_phi_N}
\end{figure}

\begin{figure}[tb]
\centering\includegraphics[width=0.5\textwidth,bb=20 200 570 610]{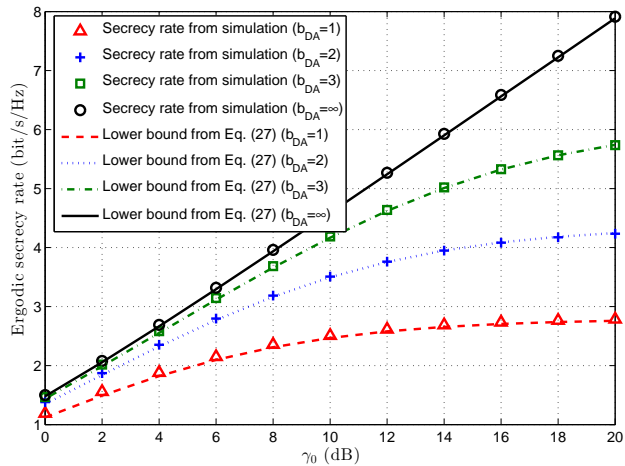}
\caption{Ergodic secrecy rate with the optimal $\phi^*$ versus SNR for null-space AN method ($N=128$, $K=8$, and $M=16$).}
\label{R_opt_phi_N}
\end{figure}

In the following, the accuracy of the obtained closed-form expressions for the approximately optimal $\phi$ are verified.
For null-space AN, Fig. \ref{R_phi_N} shows the ergodic secrecy rate with $\phi$ ranging from $0$ to $1$.
We consider 1-3 bit DACs compared with the infinite-resolution case.
Interestingly, the infinite-resolution DACs achieve the highest secrecy rate when $\phi$ is small while the lower-resolution DACs provide better rate performance for large $\phi$.
On one hand, a high DAC resolution is needed when most of the transmit power is allocated to generate AN.
On the other hand, lower-resolution DACs can achieve higher secrecy rates when most of the power is used to transmit information signals.
In fact, DAC quantization noise serves as a kind of AN to improve communication security.
The markers $\times$ in the figure denote the optimal $\phi^*$ obtained by numerical methods while the circles represent the $\phi^{\mathcal{N}*}$ in \eqref{phi_star_N}. We can see that the two match exactly.
Specifically when $\gamma_0=0$ dB, we have $\phi^{\mathcal{N}*}=0.5117, 0.3841, 0.3552, 0.3452$ for $b_\mathrm{DA}=1,2,3,\infty$, respectively.
For $\gamma_0=5$ dB, $\phi^{\mathcal{N}*}=0.5687, 0.4247, 0.3926, 0.3808$ for $b_\mathrm{DA}=1,2,3,\infty$, respectively.
For the same value of $\gamma_0$, the optimal $\phi^*$ increases with decreasing DAC resolution. This implies that more power should be allocated to the transmit signals with lower-resolution DACs.
Furthermore, Fig. \ref{R_opt_phi_N} shows the secrecy rate with the optimal $\phi^*$.
Comparing Fig. \ref{R_sec_N} with a fixed $\phi=0.8$, we see that low-resolution DACs inevitably degrade the secrecy rate regardless of the SNR.
If the optimal $\phi^*$ is achievable, higher-resolution DACs always provide more secure transmission.

\begin{figure}[tb]
\centering\includegraphics[width=0.5\textwidth,bb=20 200 570 610]{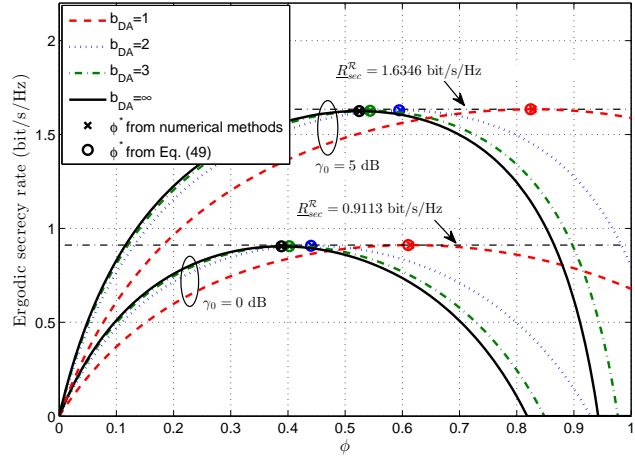}
\caption{Ergodic secrecy rate versus $\phi$ with random AN ($N=128$, $K=8$, and $M=16$).}
\label{R_phi_R}
\end{figure}

\begin{figure}[tb]
\centering\includegraphics[width=0.5\textwidth,bb=20 200 570 610]{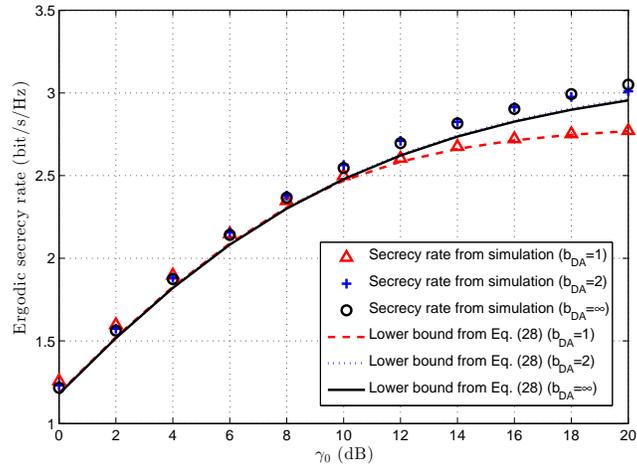}
\caption{Ergodic secrecy rate with the optimal $\phi^*$ versus SNR for random AN method ($N=128$, $K=8$, and $M=12$).}
\label{R_opt_phi_R}
\end{figure}

For random AN, Fig. \ref{R_phi_R} shows the achievable secrecy rate versus $\phi$ using low-resolution DACs.
The optimal $\phi^*$ obtained by numerical methods is denoted by $\times$ while the derived $\phi^{\mathcal{R}*}$ in \eqref{phi_star_R} is denoted by circles.
For the case that $\gamma_0=0$ dB, we have $\phi^{\mathcal{R}*}=0.6103, 0.4402, 0.4024, 0.3885$ while for $\gamma_0=5$ dB, we have $\phi^{\mathcal{R}*}=0.8243, 0.5960, 0.5435, 0.5247$, for $b_\mathrm{DA}=1,2,3,\infty$, respectively.
Unlike the results of the null-space AN method shown in Fig. \ref{R_phi_N}, the highest secrecy rates with the optimal $\phi^*$ are approximately equal for $b_{\mathrm{DA}}=1,2,3,$ and $\infty$, i.e., $\underline{R}_\mathrm{sec}^\mathcal{R}=1.6346$ and $0.9113$ for $\gamma_0=5$ and $0$~dB, respectively.
For various DAC resolutions, the same peak secrecy rate can be achieved as long as the optimal $\phi^*$ is used.
In other words, the impact of low-resolution DACs on secure transmission is insignificant.
This is because the DAC quantization noise acts as random AN at both the users and eavesdropper.
Although the quantization noise increases with a lower $b_{\mathrm{DA}}$, the same maximum secrecy rate can still be achieved by increasing $\phi$ to reduce the AN power.
Compare Fig.~\ref{R_phi_N} and Fig.~\ref{R_phi_R} and take $\gamma_0=0$ dB for instance.
When infinite-resolution DACs are deployed and using the optimal $\phi^*$, the highest secrecy rate is $\underline{R}_\mathrm{sec}^\mathcal{N}=1.4788$ with null-space AN, which is much larger than $\underline{R}_\mathrm{sec}^\mathcal{R}=0.9113$ with random AN.
When 1-bit DACs are considered, we have $\underline{R}_\mathrm{sec}^\mathcal{N}=1.1217$, which is closer to $\underline{R}_\mathrm{sec}^\mathcal{R}=0.9113$.
This implies that random AN becomes cost-efficient when low-resolution DACs are adopted. The advantage of null-space AN is marginal in this case.
The achievable secrecy rates are displayed in Fig. \ref{R_opt_phi_R} when the optimal $\phi^*$ is used.
As observed from Fig. \ref{R_phi_R}, the secrecy rates are generally not degraded by low-resolution DACs, except at high SNR with $b_{\mathrm{DA}}=1$.
Hence, using DAC resolutions beyond 1 bit is not beneficial in terms of secrecy rate.
This implies that low-resolution DACs can provide almost the same secure performance as infinite resolution DACs with random AN.
For the scenario in Fig. \ref{R_opt_phi_R} where 1-bit DACs are employed and $\gamma_0>9.8$~dB, the secrecy rate increases monotonically with $\phi\in(0,1]$ and therefore $\phi^*=1$, which is different from the cases with low SNR shown in Fig.~\ref{R_phi_R}. Under this condition, at least a two-bit DAC is needed at the BS to achieve the same secrecy rate as that in the infinite resolution case.

\section{Conclusions}
In this paper, we investigate the physical layer security of a multiuser massive MIMO system employing low-resolution DACs at the transmitter, in the presence of a passive eavesdropper.
A tight lower bound for the achievable secrecy rate of each user is derived.
We find that the DAC quantization noise can be regarded as additional AN provided by the BS and may contribute to the secure transmission.
Given a fixed power allocation factor $\phi$, low-resolution DACs can achieve superior secrecy performance under certain conditions, e.g., at low SNR or with large $\phi$.
If the optimal $\phi^*$ can be obtained, low-resolution DACs inevitably lead to secrecy rate loss with the null-space AN design method.
On the other hand, for random AN, low-resolution DACs achieve the same secrecy performance as high-resolution DACs at low SNR and thus the former are cost-efficient in this scenario.
Note that our derived results directly apply for the system with multi-antenna users if multiple data streams are transmitted to each user.
This is because in massive MIMO, an $L$-antenna user can be equivalently regarded as $L$ single-antenna users, due to the asymptotic orthogonality among channel vectors.
However, the extension becomes more complicated if a single data stream is transmitted.
Interesting future work includes further extending our current results to such a general secnario with multi-antenna users.

\begin{appendices}
\section{Proof of Lemma~\ref{lemma_SIQNR}}

In order to obtain the asymptotic expression for $\gamma_k$, we derive $S_k$, $I_k$, $Q_k$, and $A_k$ in \eqref{SIQNR3} one by one.
Consider a typical ZF-precoder under the constraint that $\textrm{tr}\{\mathbf{W}\mathbf{W}^H\}=K$, i.e.,
\begin{align}
\mathbf{W}=\sqrt{\frac{K}{\textrm{tr}\{(\mathbf{H}\mathbf{H}^H)^{-1}\}}}\mathbf{H}^H(\mathbf{H}\mathbf{H}^H)^{-1}.
\label{ZF}
\end{align}
It is well known that $\mathbf{H}\mathbf{H}^H\sim \mathcal{W}_k(N,\mathbf{I}_k)$, where $\mathcal{W}_m(n,\mathbf{\Sigma})$ denotes an $m\times m$ Wishart matrix with $n$ degrees of freedom and $\mathbf{\Sigma}$ is the covariance matrix of each column.
Assuming that $K$ and $N$ grow to infinity with a fixed ratio $\beta=\frac{K}{N}$, we have \cite{Wishart}
\begin{align}
\textrm{tr}\{(\mathbf{H}\mathbf{H}^H)^{-1}\}\xrightarrow{a.s.} \frac{\beta}{1-\beta}.
\label{tr}
\end{align}
Substituting \eqref{tr} in \eqref{ZF} yields
\begin{align}
\label{HW}
\mathbf{H}\mathbf{W}\xrightarrow{a.s.} \sqrt{K\left(\frac{1}{\beta}-1\right)}\mathbf{I}_k.
\end{align}
Thus, $S_k$ and $I_k$ converge to
\begin{align}
\label{S_asy}
S_k\xrightarrow{a.s.}(1-\rho)\phi P\left(\frac{1}{\beta}-1\right)
\end{align}
and
\begin{align}
\label{I_asy}
I_k\xrightarrow{a.s.}0,
\end{align}
respectively.

As for $Q_{k}$, the emphasis lies on the asymptotic characterizations of $\mathbf{C}_{\mathrm{DA}}$ in \eqref{cor_DA}.
For large $N$ and $K$, $\mathbf{C}_{\mathrm{DA}}$ converges to a scaled identity matrix as follows
\begin{align}
\label{C_asy}
\mathbf{C}_{\mathrm{DA}}\xrightarrow{a.s.} \rho \frac{P}{N} \mathbf{I}_N,
\end{align}
where we use \eqref{p}, \eqref{q}, and the fact that
\begin{align}
\textrm{diag}(\mathbf{W}\mathbf{W}^H)\xrightarrow{a.s.}\frac{K}{N} \mathbf{I}_N
\end{align}
and
\begin{align}
\textrm{diag}(\mathbf{V}\mathbf{V}^H)\xrightarrow{a.s.}\frac{N-K}{N} \mathbf{I}_N,
\end{align}
due to the strong law of large numbers.
Then, by substituting \eqref{C_asy} into \eqref{SIQNR3}, we have
\begin{align}
Q_k
\xrightarrow{a.s.}& \rho\frac{P}{N} \mathbf{h}_k^T\mathbf{h}_k^*
=\rho P
\label{Q_asy}
.
\end{align}

Finally for $A_{k}$, the result depends on the AN shaping matrix $\mathbf{V}$.
For the null-space AN method with $\mathbf{HV}=\mathbf{0}$, it is obvious that
\begin{align}
\label{A_asy_N}
A_k^{\mathcal{N}}=0.
\end{align}
For the random AN method, $A_k^{\mathcal{R}}$ in \eqref{SIQNR3} can be regarded as a matrix comprised of one single element, i.e., $A_k^{\mathcal{R}}\sim \mathcal{W}_1(N-K,(1-\rho)q)$, and it follows that
\begin{align}
A_k^{\mathcal{R}}&={\rm tr} \left\{A_k^{\mathcal{R}}\right\}
\nonumber
\\
&=(N-K)(1-\rho)q
\label{A_q}
\\
&=(1-\rho)(1-\phi)P,
\label{A_asy_R}
\end{align}
where \eqref{A_q} comes from the fact that ${\rm tr}\{\mathbf{A}\}=mn$ for a Wishart matrix $\mathbf{A}\sim\mathcal{W}_m(n,\mathbf{I}_m)$ \cite{Wishart},
and \eqref{A_asy_R} uses \eqref{q}.

Now, by substituting \eqref{S_asy}, \eqref{I_asy}, \eqref{Q_asy}, \eqref{A_asy_N}, and \eqref{A_asy_R} into \eqref{SIQNR3}, the asymptotic SIQNRs for the null-space and random AN methods are respectively obtained in \eqref{SIQNR_N} and \eqref{SIQNR_R}.

\section{Proof of Theorem~\ref{theorem_Ck}}

To begin with, we demonstrate that $\mathbf{X}$ defined in \eqref{X} can be approximated as a scaled Wishart matrix.
Substituting \eqref{C_asy} into \eqref{X} yields
\begin{align}
\mathbf{X}
&\xrightarrow{a.s.}(1-\rho)q  \mathbf{H}_{\mathrm{e}} \mathbf{V}\mathbf{V}^H \mathbf{H}_{\mathrm{e}}^H +\rho\frac{P}{N} \mathbf{H}_{\mathrm{e}} \mathbf{H}_{\mathrm{e}}^H
\nonumber
\\
&=(1-\rho)q  \mathbf{H}_{\mathrm{e}} \mathbf{V}\mathbf{V}^H \mathbf{H}_{\mathrm{e}}^H +\rho\frac{P}{N} \mathbf{H}_{\mathrm{e}} [\mathbf{V}~\mathbf{V}_0] [\mathbf{V}~\mathbf{V}_0]^H \mathbf{H}_{\mathrm{e}}^H
\label{X4}
\\
&=\left[(1-\rho)q +\rho\frac{P}{N}\right]
\underbrace{\mathbf{H}_1 \mathbf{H}_1^H}_{\mathbf{W}_1}
 +\rho\frac{P}{N}
\underbrace{\mathbf{H}_2 \mathbf{H}_2^H}_{\mathbf{W}_2}
\label{X5}
,
\end{align}
where \eqref{X4} uses the fact that $[\mathbf{V}~\mathbf{V}_0] [\mathbf{V}~\mathbf{V}_0]^H=\mathbf{I}_M$ since $[\mathbf{V}~\mathbf{V}_0]$ is a complete orthogonal basis with dimension $N$,
and \eqref{X5} utilizes the definitions $\mathbf{H}_1\triangleq \mathbf{H}_{\mathrm{e}} \mathbf{V}$ and $\mathbf{H}_2\triangleq \mathbf{H}_{\mathrm{e}} \mathbf{V}_0$.
From \eqref{X5}, $\mathbf{X}$ is statistically equivalent to a weighted sum of two scaled Wishart matrices, i.e., $\mathbf{X}_1\sim \mathcal{W}_M(N-K,\mathbf{I}_M)$ and $\mathbf{X}_2\sim \mathcal{W}_M(K,\mathbf{I}_M)$.
Strictly speaking, $\mathbf{X}$ is not a Wishart matrix and the exact distribution of $\mathbf{X}$ is intractable.
However, $\mathbf{X}$ may be accurately approximated as a single scaled Wishart matrix, $\mathbf{X}\sim \mathcal{W}_M(\eta,\lambda\mathbf{I}_M)$, where the parameters $\eta$ and $\lambda$ are chosen such that the first two moments of $\mathbf{X}$ and $\left[(1-\rho)q +\rho\frac{P}{N}\right]\mathbf{W}_1+\rho\frac{P}{N} \mathbf{W}_2$ are identical \cite{Secure_Zhu_1}, which yields
\begin{align}
\label{term1}
\eta\lambda&=(N-K)\left[(1-\rho)q +\rho\frac{P}{N}\right]+K\rho\frac{P}{N}
\end{align}
and
\begin{align}
\label{term2}
\eta\lambda^2=(N-K)\left[(1-\rho)q +\rho\frac{P}{N}\right]^2+K\left(\rho\frac{P}{N}\right)^2.
\end{align}
Substituting \eqref{q} into \eqref{term1} and \eqref{term2}, $\eta$ and $\lambda$ are obtained as
\begin{align}
\eta=N\frac{[(1-\rho)(1-\phi)+\rho]^2}{[(1-\rho)(1-\phi)+\rho]^2+(1-\rho)^2(1-\phi)^2\frac{K}{N-K}}
\label{eta}
\end{align}
and
\begin{align}
\lambda= \frac{P}{N}\frac{[(1-\rho)(1-\phi)+\rho]^2+(1-\rho)^2(1-\phi)^2\frac{K}{N-K}}{(1-\rho)(1-\phi)+\rho},
\label{lambda}
\end{align}
respectively.

Next, we apply Jensen's inequality which yields an upper bound for the eavesdropper's capacity:
\begin{align}
C_k
&\leq \log_2\left[1+ (1-\rho)p~\mathbb{E}\left\{\mathbf{w}_k^H  \mathbf{H}_{\mathrm{e}}^H \mathbf{X}^{-1} \mathbf{H}_{\mathrm{e}} \mathbf{w}_k \right \}\right]
\nonumber
\\
&=\log_2\left[1+ \frac{(1-\rho)p}{\lambda(\eta-M)}~\mathbb{E}\left\{\mathbf{w}_k^H  \mathbf{H}_{\mathrm{e}}^H \mathbf{H}_{\mathrm{e}}  \mathbf{w}_k \right \}\right]
\label{Wishsrt_inv}
\\
&=\log_2\left[1+ \frac{(1-\rho)pM}{\lambda(\eta-M)}~\mathbb{E}\left\{\mathbf{w}_k^H \mathbf{w}_k \right \}\right]
\label{Central}
\\
&=\log_2\left[1+ \frac{(1-\rho)p M}{\lambda(\eta-M)} \right]
\label{C_k}
,
\end{align}
where \eqref{Wishsrt_inv} utilizes the property that $\mathbf{A}^{-1} \xrightarrow{a.s.} \frac{1}{n-m} \mathbf{I}_m$ for a Wishart matrix $\mathbf{A}\sim\mathcal{W}_m(n,\mathbf{I}_m)$ with $n>m$ \cite{DAC1},
\eqref{Central} uses the fact that $\frac{1}{M}\mathbf{H}_{\mathrm{e}}^H \mathbf{H}_{\mathrm{e}}-\mathbf{I}_N \xrightarrow{a.s.} \mathbf{0}_N$ due to the Central Limit Theorem,
and \eqref{C_k} applies the weak law of large numbers and $\mathbb{E}\{\mathbf{w}_k^H \mathbf{w}_k\}=\frac{1}{K}\sum\limits_{k=1}^K\mathbf{w}_k^H\mathbf{w}_k=\frac{1}{K}\mathrm{tr}\{\mathbf{W}^H\mathbf{W}\}=1$.
Note that the derivation in \eqref{Wishsrt_inv} only holds for an invertible $\mathbf{X}\sim\mathcal{W}_M(\eta,\lambda\mathbf{I}_M)$, which yields $\eta-M>0$.
By substituting \eqref{alpha}, \eqref{beta}, and \eqref{eta}, we have
\begin{align}
&\eta-M
\nonumber
\\
=&N\frac{(1\!-\!\alpha)(1\!-\!\beta)[(1-\rho)(1-\phi)+\rho]^2\!-\!\alpha\beta(\!-\!\rho)^2(1\!-\!\phi)^2}{(1-\beta)[(1-\rho)(1-\phi)+\rho]^2+\beta(1-\rho)^2(1-\phi)^2}
\nonumber
\\
=&N\frac{(1-\alpha)(1-\beta)(1-\rho)^2(1-\phi)^2}{(1-\beta)[(1-\rho)(1-\phi)+\rho]^2+\beta(1-\rho)^2(1-\phi)^2}
\left[ \left(\!1\!+\!\frac{\rho}{(1\!-\!\rho)(1\!-\!\phi)}\right)^2\!-\!\frac{\alpha\beta}{(1\!-\!\alpha)(1\!-\!\beta)}\right]
\nonumber
\\
>&0
.
\label{eta2}
\end{align}
Regardless of the values of $\rho\in(0,1)$ and $\phi\in(0,1]$, \eqref{eta2} holds if $\frac{\alpha\beta}{(1-\alpha)(1-\beta)}<1$, which yields $\alpha+\beta<1$ with $\beta\in(0,1)$ and $\alpha\in(0,1)$.
Fortunately, this is a common condition for massive MIMO systems with large $N$.

Finally by substituting \eqref{p}, \eqref{eta}, and \eqref{lambda} into \eqref{C_k}, the upper bound in \eqref{Ck_bound} is directly obtained.

\end{appendices}

\end{document}